\ifdefined\lipics{}
\documentclass[a4paper,UKenglish]{lipics-v2018}
\else
\documentclass[11pt]{article}

\usepackage[margin=1in]{geometry}
\usepackage{amsmath}
\usepackage{amsthm}
\newtheorem{theorem}{Theorem}[section]
\newtheorem{lemma}[theorem]{Lemma}
\newtheorem{corollary}[theorem]{Corollary}
\newtheorem{claim}[theorem]{Claim}
\newtheorem{definition}[theorem]{Definition}

\usepackage{amsfonts}
\usepackage{amssymb}
\usepackage{url}
\fi
\usepackage[utf8]{inputenc}
\usepackage[T1]{fontenc}
\usepackage{lmodern}

\usepackage{graphicx,comment}
\usepackage[linesnumbered,noresetcount,vlined]{algorithm2e}
\usepackage{environ,tabularx}

\usepackage{xcolor}
\usepackage[inline]{enumitem}

\newenvironment{lemma-repeat}[1]{\begin{trivlist}
\item[\hspace{\labelsep}{\bf\noindent Lemma \ref{#1} }]\em }%
{\end{trivlist}}
\newenvironment{theorem-repeat}[1]{\begin{trivlist}
\item[\hspace{\labelsep}{\bf\noindent Theorem \ref{#1} }]\em }%
{\end{trivlist}}

\newcommand*\samethanks[1][\value{footnote}]{\footnotemark[#1]}

\ifdefined\noComments
\newcommand{\Gcomment}[1]{}
\newcommand{\Ecomment}[1]{}
\newcommand{\Rcomment}[1]{}
\else
\newcommand{\Gcomment}[1]{{\color{red} Gregory: #1}}
\newcommand{\Ecomment}[1]{{\color{red} Guy: #1}}
\newcommand{\Rcomment}[1]{{\color{green} Ran: #1}}

\fi

\newcommand{\eps}{\varepsilon}
\newcommand{\congest}{\textsc{congest}\xspace}
\newcommand{\local}{\textsc{local}\xspace}
\newcommand{\WHVC}{\textsc{mwhvc}\xspace}

\newcommand{\MWHVC}{\textsc{mwhvc}\xspace}
\newcommand{\MWVC}{\textsc{mwvc}\xspace}
\newcommand{\opt}{\textsf{opt}}
\DeclareMathOperator*{\argmin}{argmin} % thin space, limits underneath in displays
\DeclareMathOperator*{\alg}{\textup{\MWHVC\xspace}} 
\DeclareMathOperator*{\deal}{deal} 
\newcommand{\parentheses}[1]{\left(#1\right)}
\newcommand{\logp}[1]{\log\parentheses{#1}}

\newcommand{\set}[1]{\left\{ #1 \right\}}

\makeatletter
\newcommand{\problemtitle}[1]{\gdef\@problemtitle{#1}}% Store problem title
\newcommand{\probleminput}[1]{\gdef\@probleminput{#1}}% Store problem input
\newcommand{\problemoutput}[1]{\gdef\@problemoutput{#1}}% Store problem output
\newcommand{\problemobj}[1]{\gdef\@problemobj{#1}}% Store problem question
\NewEnviron{problem}{
  \problemtitle{}\probleminput{}\problemoutput{}\problemobj{}% Default input is empty
  \BODY% Parse input
  \par\addvspace{.5\baselineskip}
  \noindent
  \begin{tabularx}{\textwidth}{@{\hspace{\parindent}} l X c}
    \textbf{Input:} & \@probleminput \\% Input
    \textbf{Output:} & \@problemoutput \\% Input
    \textbf{Objective:} & \@problemobj% Question
  \end{tabularx}
  \par\addvspace{.5\baselineskip}
}
\makeatother

\usepackage{microtype}%if unwanted, comment out or use option "draft"

\title{Optimal Distributed Weighted Set Cover Approximation }

\ifdefined\lipics{}
\titlerunning{$(f+\epsilon)$-Approximation for Weighted Hypergraph Vertex Cover in $O\left(\log{\Delta}\right)$ Rounds}%optional, please use if title is longer than one line
\fi

\author{Ran Ben-Basat\thanks{Technion,  \texttt{sran@cs.technion.ac.il}}
\and Guy Even\thanks{Tel Aviv University, \texttt{guy@eng.tau.ac.il}}
			\and Ken-ichi Kawarabayashi \thanks{NII, Japan, \texttt{k\_keniti@nii.ac.jp}, \texttt{greg@nii.ac.jp}} \and Gregory Schwartzman\samethanks}
\ifdefined\lipics{}
\author{Ran Ben Basat}{Technion, Haifa, Israel}{sran@cs.technion.ac.il}{}{}%mandatory, please use full name; only 1 author per \author macro; first two parameters are mandatory, other parameters can be empty.

\author{Guy Even}{Tel Aviv University, Tel Aviv, Israel}{guy@eng.tau.ac.il}{}{}

\author{Ken-ichi Kawarabayashi}{NII, Tokyo, Japan}{k\_keniti@nii.ac.jp}{}{}

\author{Gregory Schwartzman}{NII, Tokyo, Japan}{greg@nii.ac.jp}{}{}

\authorrunning{Ran Ben Basat, Guy Even, Ken-ichi Kawarabayashi, and Gregory Schwartzman}%mandatory. First: Use abbreviated first/middle names. Second (only in severe cases): Use first author plus 'et al.'

\Copyright{Ran Ben Basat, Guy Even, Ken-ichi Kawarabayashi, and Gregory Schwartzman}%mandatory, please use full first names. LIPIcs license is "CC-BY";  http://creativecommons.org/licenses/by/3.0/

\subjclass{Theory of computation$\rightarrow$Distributed algorithms}% mandatory: Please choose ACM 2012 classifications from https://www.acm.org/publications/class-2012 or https://dl.acm.org/ccs/ccs_flat.cfm . E.g., cite as "General and reference $\rightarrow$ General literature" or \ccsdesc[100]{General and reference~General literature}. 

%mandatory

\category{}%optional, e.g. invited paper

\EventEditors{}
\EventNoEds{4}
\EventLongTitle{}
\EventShortTitle{DISC 2018}
\EventAcronym{DISC}
\EventYear{2018}
\EventDate{}
\EventLocation{}
\EventLogo{eatcs}
\SeriesVolume{}
\ArticleNo{}
\fi
\begin{document}
\maketitle

\begin{abstract}
We present a time-optimal deterministic distributed algorithm for approximating a minimum weight vertex cover in hypergraphs of rank $f$.
This problem is equivalent to the Minimum Weight Set Cover Problem in which the frequency of every element is bounded by $f$.
The approximation factor of our algorithm is $(f+\epsilon)$.
Let $\Delta$ denote the maximum degree in the hypergraph.
Our algorithm runs in the \congest model and requires $O(\log{\Delta} / \log \log \Delta)$ rounds, for constants $\epsilon \in (0,1]$ and $f\in\mathbb N^+$.
This is the first distributed algorithm for this problem whose running time does not depend on the vertex weights or the number of vertices. Thus adding another member to the exclusive family of \emph{provably optimal} distributed algorithms.

For constant values of $f$ and $\epsilon$, our algorithm improves over the $(f+\epsilon)$-approximation algorithm of \cite{KuhnMW06} whose running time is $O(\log \Delta + \log W)$, where $W$ is the ratio between the largest and smallest vertex weights in the graph. 

\end{abstract}
\thispagestyle{empty}
\newpage
\setcounter{page}{1}

\section{Introduction}
In the Minimum Weight Vertex Cover (\MWVC) problem, we are given an undirected graph $G=(V,E)$ with vertex weights $w:V\to \set{1,\ldots,W}$, for $W=n^{O(1)}$. The goal is to find a minimum weight \emph{cover} $U\subseteq V$ such that $\forall e\in E: e\cap U\neq\emptyset$. This problem is one of the classical NP-hard \mbox{problems presented in~\cite{Karp72}.}

In this paper, we consider the Minimum Weight Hypergraph Vertex Cover (\MWHVC) problem, a generalization of the \MWVC problem to \emph{hypergraphs} of rank $f$. In a hypergraph, $G=(V,E)$, each edge is a nonempty subset of the vertices. A hypergraph is of rank $f$ if the size of every hyperedge is bounded by $f$. 
The \MWVC problem naturally extends to \MWHVC using the above definition. Note that \MWHVC is equivalent to the Minimum Weight Set Cover problem with element frequencies bounded by $f$.

We consider the \MWHVC problem in the distributed setting, where the communication network is a bipartite graph $H(E\cup V,\set{\set{e,v} \mid v\in e})$. We refer to the network vertices as \emph{nodes} and network edges as \emph{links}. The nodes
of the network are the hypergraph vertices on one side and hyperedges on the other side. There is a network link between vertex $v\in V$ and hyperedge $e\in E$ iff $v\in e$. The computation is performed in synchronous rounds, where messages are sent between neighbors in the communication network. As for message size, we consider the \congest model where message sizes are bounded to $O(\log|V|)$. This is more
restrictive than the \local model \mbox{where message sizes are unbounded.}

Denoting by $\Delta$ the maximum vertex degree in $G$, any distributed constant-factor approximation algorithm requires
$\Omega(\log\Delta / \log\log \Delta)$ rounds to terminate, even for unweighted graphs and $f=2$~\cite{KuhnMW16}. Two results match the lower bound. For the Minimum Weight Vertex Cover Problem in graphs $(f=2)$, the lower bound was matched by~\cite{Bar-YehudaCS17} with a $(2+\epsilon)$-approximation algorithm (BCS algorithm) with optimal round complexity for every $\epsilon=\Omega(\log\log\Delta / \log \Delta)$.  \footnote{Recently,   the range of $\epsilon$ for which the runtime is optimal was improved to $\Omega(\log^{-c}\Delta)$ for any  $c=O(1)$~\cite{BEKS18}.} The progress of the BCS algorithm is analyzed via a trade-off between reducing the degree of the vertices and reducing the weight of the vertices. We do not know how to generalize the BCS algorithm and its analysis to hypergraphs. For the Minimum Cardinality Vertex Cover in Hypergraphs Problem, the lower bound was matched by ~\cite{unweightedHVC} with an $(f+\eps)$-approximation algorithm.  The round complexity in~\cite{unweightedHVC} is
$O\parentheses{f/\epsilon \cdot \frac{\log(f\cdot\Delta)}{\log\log(f\cdot\Delta)}}$, which is optimal for constant $f$ and $\eps$. The algorithm in~\cite{unweightedHVC} and its analysis is a deterministic version of the maximal independent set algorithm
of~\cite{ghaffari2016improved}. We do not know how to generalize the algorithm in~\cite{unweightedHVC} and its analysis to hypergraphs with vertex weights.

In this paper, we present a deterministic distributed $(f+\epsilon)$-approximation algorithm for minimum weight vertex cover
in $f$-rank hypergraphs, which completes in
$O(\log{\Delta} / \log \log \Delta)$ rounds in the \congest model, for any constants $\epsilon \in (0,1)$ and $f\in\mathbb N^+$. Our running
time is optimal according to a lower bound by
\cite{KuhnMW16}.%
\footnote{The dependence of the round complexity of our algorithm on $\epsilon$ and  $f$ is given by $O\parentheses{
    \frac{\log \Delta}{\log\log\Delta}
    +
    \parentheses{\frac{f^2}{\eps}}^{1/\gamma} \cdot \log\log \Delta
    }
  $, for every constant $\gamma\in(0,1)$ (Theorem~\ref{thm:refined}).
  }
%Further, for any constant $f,\epsilon$ our algorithm terminates in the optimal $O(\log{\Delta} / \log \log \Delta)$. 
%That is, our algorithm has
%optimal round complexity even for certain sub-constant values of
%$\epsilon$.  
This is the first distributed algorithm for this problem whose round complexity does not depend on the node weights. For constant values of $f$, Astrand et al.~\cite{AstrandS10} present an $f$-approximation algorithm whose running time is $O(\Delta^2 + \Delta\cdot\log^* W)$ ($W$ is the ratio between the largest and smallest weights in the graph).  Kuhn et al.~\cite{Kuhn2005,KuhnMW06} present an $(f+\epsilon)$-approximation algorithm that terminates in $O(\log \Delta + \log W)$ rounds. To the best of our knowledge, these are the only works that deal with the Minimum Weight Hypergraph Vertex Cover Problem (\MWHVC) in the
distributed setting.

Our algorithm is one of a handful of distributed algorithms for \emph{local} problems which are \emph{provably optimal} \cite{Bar-YehudaCGS17, Bar-YehudaCS17, ChangKP16, ColeV86, GhaffariS17, unweightedHVC}. Among these are the classic Cole-Vishkin algorithm \cite{ColeV86} for 3-coloring a ring, the more recent results of \cite{Bar-YehudaCGS17} and \cite{Bar-YehudaCS17} for \MWVC and Maximum Matching, and the very recent result of \cite{unweightedHVC} for Minimum Cardinality Hypergraph Vertex Cover.

\subsection{Tools and techniques} 
Our solution employs the Primal-Dual schema. The Primal-Dual approach introduces, for every hyperedge $e \in E$, a dual variable denoted by $\delta(e)$. The dual edge packing constraints are $\forall v\in V,\sum_{v \in e} \delta(e) \leq w(v)$. If for some $\beta \in [0,1)$ it holds that $\sum_{v \in e} \delta(e) \geq (1-\beta)\cdot w(v)$, we say the $v$ is $\beta$-tight. Let $\beta=\eps/(f+\eps)$. For every feasible dual solution, the weight of the set of $\beta$-tight vertices is at most $(f+\eps)$ times the weight of an optimal (fractional) solution. The algorithm terminates when the set of $\beta$-tight edges constitutes a vertex cover.

The challenge in designing a distributed algorithm is in controlling the rate at which we increase the dual variables. On the one hand, we must grow them rapidly to reduce the number of communication rounds. On the other hand, we may not violate the edge packing constraints. 
The algorithm proceeds in iterations, each of which requires a constant number of communication rounds. We initialize the dual variables in a "safe" way so that feasibility is guaranteed. We refer to the additive increase of the dual variable $\delta(e)$ in iteration $i$ by $\deal_i(e)$. Loosely speaking, the algorithm  increases the increments $\deal_i(e)$ exponentially (multiplication by $\alpha$) provided that no vertex $v\in e$ is $(\beta / \alpha)$-tight with respect to the deals of the previous iteration. Otherwise, the increment $\deal_i(e)$ equals the previous increment $\deal_{i-1}(e)$. 
The analysis builds on two observations: (1)~The number of times that the increment $\deal(e)$ is multiplied by $\alpha$ is bounded by $\log_\alpha \Delta$. (2)~The number of iterations in which a vertex is $(\beta/\alpha)$-tight with respect to the deals of the previous iteration is at most $\alpha/\beta$. Hence the total number of iterations is bounded by $\log_\alpha \Delta+f\cdot \alpha/\beta$. Setting $\alpha=\log \Delta/\log\log \Delta$ implies that the number of iterations is  $O(\log \Delta / \log \log \Delta)$.

\section{Problem Formulation}
Let $G=(V,E)$ denote a hypergraph. Vertices in $V$ are equipped with nonnegative weights $w(v)$.  For a subset $U\subseteq V$, let $w(U)\triangleq \sum_{v\in U} w(v)$.  Let $E(U)$ denote the set of hyperedges that are incident to some vertex in $U$ (i.e., $E(U)\triangleq \{ e\in E \mid e\cap U \neq \emptyset\}$).

\medskip
\noindent
The \emph{Minimum Weight Hypergraph Vertex Cover} Problem (\WHVC) is
defined as follows.
\begin{problem}
%  \problemtitle{Minimum Weight Hypergraph Vertex Cover}% (\PCWHVC)}%
  \probleminput{Hypergraph $G=(V,E)$ with vertex weights $w(v)$.}%
  \problemoutput{A subset $C\subseteq V$ such that $E(C)=E$. }%
  \problemobj{Minimize $w(C)$.}
\end{problem}
 
The \WHVC Problem is equivalent to the Weighted Set Cover Problem.
Consider a set system $(X,\mathcal{U})$, where $X$ denotes a set of elements and $\mathcal{U}=\set{U_1,\ldots, U_m}$ denotes a collection of subsets of $X$. 
The reduction from the set system $(X,\mathcal{U})$ to a hypergraph $G=(V,E)$ proceeds as follows.
The set of vertices is $V\triangleq \set{u_1,\ldots, u_m}$ (one vertex $u_i$ per subset $U_i$).
The set of edges is $E\triangleq\set{e_x}_{x\in X}$ (one hyperedge $e_x$ per element $x$), where $e_x\triangleq \set{u_i: x\in U_i}$.
The weight of vertex $u_i$ equals the weight of the subset $U_i$.

\section{Distributed $(f+\eps)$-Approximation Algorithm for {\small MWHVC}}
\label{sec: alg}
\subsection{Input}
The input is a hypergraph $G=(V,E)$ with non-negative vertex weights $w:V\rightarrow \mathbb{R}^+$ and an approximation ratio parameter $\eps\in (0,1]$.  
We denote the rank of $G$ by $f$ (i.e., each hyperedge contains at most $f$ vertices) and the maximum degree of $G$ by $\Delta$ (i.e., each vertex belongs to at most $\Delta$ edges).

\paragraph{Assumptions.}
We assume that 
\begin{enumerate*}[label={(\roman*)}]
\item Vertex weights are polynomial in $n=|V|$ so that sending a vertex weight requires $O(\log n)$ bits. 
\item Vertex degrees are polynomial in $n$ (i.e., $|E(v)|=n^{O(1)}$) so that sending a vertex degree requires $O(\log n)$ bits. 
Since $|E(v)|\leq n^f$, this assumption trivially holds for constant $f$.
\item The maximum degree is at least $3$ so that $\log \log \Delta >0$.
\end{enumerate*}

\subsection{Output}
A vertex cover $C\subseteq V$. Namely, for every hyperedge $e\in E$, the intersection $e\cap C$ is not empty.  The set $C$ is maintained locally in the sense that every vertex $v$ knows whether it belongs to $C$ or not.

\subsection{Communication Network} 
The communication network $N(E\cup V,\set{\set{e,v} \mid v\in e})$ is a bipartite graph.  There are two types of nodes in the network: \emph{servers} and \emph{clients}. The set of servers is $V$ (the vertex set of $G$) and the set of clients is $E$ (the hyperedges in $G$).  There is a link $(v,e)$ from server $v\in V$ to a client $e\in E$ if $v\in e$.  We note that the degree of the clients is
bounded by $f$ and the degree of the servers is bounded by $\Delta$.

\subsection{Parameters and Variables}
\begin{itemize}
\item The approximation factor parameter $\eps\in(0,1]$ and the rank $f$
  determine the parameter $\beta$ defined by
  $\beta\triangleq\epsilon/(f+\eps)$.
\item The parameter $\alpha$ is set to $\log \Delta/\log\log \Delta$
  and determines the factor by which ``deals'' are multiplied\footnote{ 
    For simplicity, we assume that $\Delta$ is known and that
    $\Delta\ge 3$.  The assumption that the maximal degree $\Delta$ is
    known to all vertices is not required. Instead, each hyperedge $e$
    can compute a local maximum degree $\Delta(e)$, where
    $\Delta(e)\triangleq\max_{u\in e}|E(u)|$.  The local maximum degree
    $\Delta(e)$ can be used instead of $\Delta$ to define local value
    of the multiplier $\alpha=\alpha(e)$.
}.
See Section~\ref{sec:refine} in the Appendix for a setting of $\alpha$ that reduces the dependency of the running time on $\eps $ and $f$.
  
\begin{comment}
  \begin{align*}
    \alpha&\triangleq
\begin{cases}
\frac{\beta\cdot\log \Delta}{\log\log \Delta} & \mbox{if $\beta>\parentheses{\frac{\log\log\Delta}{\log\Delta}}^{0.999}$}\\
2&\mbox{otherwise}.
\end{cases}
\end{align*}
\end{comment}
\item We denote the dual variables at the end of iteration $i$ by
  $\delta_i(e)$ (see Appendix~\ref{sec:primal dual} for a description
  of the dual edge packing linear program).  The amount by which
  $\delta_i(e)$ is increased in iteration $i$ is denoted by
  $\deal_i(e)$. Namely, $\delta_i(e)=\sum_{j\leq i}\deal_j(e)$.
\end{itemize}

\subsection{Notation}
\begin{itemize}
\item We say that an edge $e$ is \emph{covered} by $C$ if
  $e\cap C\neq\emptyset$.
%\item Each vertex $v$ has a weight $w(v)$.
  %Following the BCS-algorithm, we refer to $\beta\cdot w(v)$ as the \emph{vault} of $v$.
\item Let $E(v)\triangleq \{ e\in E\mid v\in e\}$ denote the set of
  hyperedges that contain $v$.
\item For every vertex $v$, the algorithm maintains a subset
  $E'(v)\subseteq E(v)$ that consists of the uncovered hyperedges in
  $E(v)$ (i.e., $E'(v)=\{e\in E(v)\mid e\cap C=\emptyset\}$).
\end{itemize}
\subsection{Algorithm {\sc MWHVC}}
\begin{enumerate}
\item Initialization. Set $C\gets \emptyset$. For every vertex $v$, set $E'(v)\gets E(v)$. 
\item Iteration $i=0$. The edge $e$ collects the  weight $w(v)$  and degree $|E(v)|$
  from every vertex $v\in e$, and sets:
  $\deal_0(e)=\beta\cdot \min_{v\in e}\{ w(v)/|E(v)|$\}. The value
  $\deal_0(e)$ is sent to every $v\in e$. The dual variable is updated
  $\delta_0(e)\gets \deal_0(e)$.
\item For $i=1$ to $\infty$ do:
  \begin{enumerate}
  \item\label{item:tight}
    Every vertex $v\not\in C$ checks if it is $\beta$-tight.  If
    $\sum_{e\in E(v)} \delta_{i-1}(e) \geq (1-\beta)\cdot w(v)$, then $v$
    joins the cover $C$, sends a message to every $e\in E'(v)$
    that $e$ is covered, and ($v$) terminates.
  \item For every uncovered edge $e$, if $e$ receives a message that it is
    covered, then it tells all its vertices that $e$ is covered, and $e$ terminates.
  \item For every vertex $v\notin C$, if it receives a message from $e$ that
    $e$ is covered, then $E'(v)\gets E'(v)\setminus \{e\}$.
    If $E'(v)=\emptyset$, then $v$ terminates  (without joining the cover).
  \item \label{item:good}
    For every vertex $v\notin C$, if
    $\sum_{e\in E'(v)} \deal_{i-1}(e) \leq \frac{\beta}{\alpha}\cdot
    w(v)$, then send the message ``raise'' to every $e\in E'(v)$, else
    send the message ``stuck'' to every $e\in E'(v)$.
  \item \label{item:update}
    For every uncovered edge $e$. If $e$
    received a ``stuck'' message then $\deal_i(e)\gets\deal_{i-1}(e)$,
    else (if all incoming messages are ``raise'')
    $\deal_i(e)\gets\alpha\cdot \deal_{i-1}(e)$. Send $\deal_i(e)$ to
    every $v\in e$, who updates $\delta_i(e)\gets \delta_{i-1}(e)+\deal_i(e)$.
  \end{enumerate}
\end{enumerate}
%%%%%%%%%%%%%%%%%%%%%%%%%%%%%%%%%%%%%%%%%%%%%

\paragraph*{Termination}
Every vertex $v$ terminates when either $v\in C$ or every edge $e\in E(v)$ is covered (i.e., $E'(v)=\emptyset$).
Every edge $e$ terminates when it is covered (i.e., $e\cap C\neq\emptyset$).  

\paragraph*{Execution in CONGEST.}
See Section~\ref{sec:congest} in the Appendix for a discussion of how
Algorithm \WHVC is executed in the \congest model.

%%%%%%%%%%%%%%%%%%%%%%%%%%%%%%%%%%%%%%%%%%%%%%%%%%%%
\begin{comment}

\section{check list}
\begin{itemize}
  \item fine tune complexity - tomorrow
  \item prize collecting version - not for SOSA
\item
dependency on $\eps$ - not for SOSA?
\item
size of messages in congest - done.
\item
approx the fractional LP - not for SOSA
\end{itemize}

\end{comment}
%%%%%%%%%%%%%%%%%%%%%%%%%%%%%%%%%%%%%%%%%%%%%%%%%%%%%%%%%%%%%%%%%%%
\section{Algorithm Analysis}
\subsection{Approximation Ratio}
The following claim states that, in each iteration, the sum of the deals
of edges incident to a vertex $v$ is bounded by $\beta\cdot w(v)$.
\begin{claim}\label{claim:vault}
  If $v\not\in C$, then
  $\sum_{e\in E'(v)} \deal_i (e) \leq \beta\cdot w(v)$.
\end{claim}
\begin{proof}
  The proof is by induction on $i$.  The induction basis, for $i=0$,
  holds because $\deal_0(e)\leq \beta\cdot w(v)/|E(v)|$ for every edge
  $e\in v$.   The induction step, for
  $i\geq 1$, considers two cases.  If $\deal_{i+1}(e)=\deal_i(e)$ for
  every $e\in E'(v)$, then the induction step follows from the
  induction hypothesis. If there exists an edge $e\in E'(v)$ such that
  $\deal_{i+1}(e)=\alpha\cdot\deal_i(e)$, then 
  Step~\ref{item:good} implies that
  $\sum_{e\in E'(v)} \deal_{i}(e)\leq \alpha \cdot \sum_{e\in E'(v)}
  \deal_{i-1}(e)\leq \beta\cdot w(v)$, as required.
\end{proof}

\medskip\noindent If an edge $e$ is covered in iteration $j$, then
$e$ terminates and $\delta_i(e)$ is not set for $i\geq j$. In this case, we define $\delta_i(e)=\delta_{j-1}(e)$, namely, the last value assigned to a dual variable. 
\begin{claim}\label{claim:feasible}
  For every $i\geq 0$ the dual variables $\delta_i(e)$ constitute a
  feasible edge packing. Namely,
  \begin{align*}
    \sum_{e\in E(v)} \delta_i(e) &\leq w(v) &\text{for every vertex $v\in V$,}\\
    \delta_i(e)&\geq 0&\text{for every edge $e\in E$.}
  \end{align*}
\end{claim}
\begin{proof}
  Nonnegativity follows from the initialization and the positive increases by deals.  The packing constraints are proved by induction
  on the number of iterations.  The induction basis, for $i=0$, holds because
  $\sum_{e\in E(v)} \delta_0(e) = \sum_{e\in E(v)} \deal_0(e) =\beta\cdot w(v)$. (Recall that $\beta=\eps/(f+\eps)<1$.)  The
  induction step is proved as follows. By Step~\ref{item:update}, if $e\in E'(v)$, then $\delta_{i}(e)=\delta_{i-1}(e)+\deal_i(e)$,
  otherwise $\delta_i(e)=\delta_{i-1}(e)$.  By Step~\ref{item:tight},
  $\sum_{e\in E(v)} \delta_{i-1}(e)< (1-\beta)\cdot w(v)$.  By Claim~\ref{claim:vault} in Appendix~\ref{sec:primal dual},
  $\sum_{e\in E'(v)} \deal_i(e)\leq \beta\cdot w(v)$, and the claim follows.
  \end{proof}

\medskip\noindent  Let $\opt$ denote the cost of an optimal (fractional) weighted
  vertex cover of $G$.
  \begin{corollary}
Upon termination, the approximation ratio of Algorithm \WHVC is $f+\eps$.
\end{corollary}
\begin{proof}
  Throughout the algorithm, the set $C$ consists of $\beta$-tight
  vertices. By Claim~\ref{claim:ratio PCWHVC}, $w(C)\leq (f+\eps)\cdot \opt$. Upon
  termination, $C$ constitutes a vertex cover, and the corollary
  follows.
\end{proof}

%%%%%%%%%%%%%%%%%%%%%%%%%%%%%%%%%%%%%%%%%%%%%%%% 
\subsection{Communication Rounds Analysis}\label{sec:runtime}
In this section, we prove that the number of communication rounds of
Algorithm~\MWHVC{}\ is bounded by
$ O\parentheses{\frac{f^2\cdot \log \Delta}{\eps\cdot \log\log \Delta}}$.
It suffices to
bound the number of iterations because each iteration consists of a constant number of 
communication rounds.

%%%%%%%%%%%%%%%%%%%%%%%%%%
\subsubsection{Raise or Stuck Iterations}

\begin{definition}
  An iteration $i\geq 1$ is an $e$-\emph{raise} iteration if $\deal_i(e)=\alpha\cdot\deal_{i-1}(e)$.
  An iteration $i$ is a $v$-\emph{stuck} iteration if
$v$ sent the message ``stuck'' in iteration $i$.
\end{definition}
Note that if iteration $i$ is a $v$-stuck iteration and $v\in e$, then $\deal_i(e)=\deal_{i-1}(e)$ and $i$ is not an $e$-raise iteration.

\medskip\noindent
We bound the number of $e$-raise iterations as~follows.
\begin{lemma}\label{lemma:regular passes}
The number of $e$-raise iterations is bounded by $\log_\alpha \Delta\;.$
\end{lemma}
\begin{proof}
  Let $v^*$ denote a vertex with minimum normalized weight in $e$.
  The first deal satisfies
  $\deal_0(e)= \beta\cdot w(v^*)/|E(v^*)| \geq \beta\cdot
  w(v^*)/\Delta$.  By Claim~\ref{claim:vault},
   $\deal_i(e)\leq \beta\cdot w(v^*)$. Since the deal is multiplied by
  $\alpha$ in each $e$-raise iteration, the lemma follows.
\end{proof}
\medskip\noindent
We bound the number of $v$-stuck as~follows.
\begin{lemma}\label{lemma:good}
The number of $v$-stuck iterations is bounded by $\frac{\alpha}{\beta}\;.$
\end{lemma}
\begin{proof}
  Suppose that iteration $i$ (for $i\geq 1$) is a $v$-stuck iteration.
  This implies that $\sum_{e\in E'(v)}\deal_{i-1}(e)> \frac{\beta}{\alpha}\cdot w(v)$.
  Thus $\delta_{i-1}(e)-\delta_{i-2}(e)> \frac{\beta}{\alpha}\cdot w(v)$.
  Had there been more than $\frac{\alpha}{\beta}$ iterations that are $v$-stuck, then the dual variable $\delta(e)$ would be larger than $w(v)$, contradicting Claim~\ref{claim:feasible}. 
\end{proof}
\medskip\noindent

\subsubsection{Putting it Together}
%Set $\alpha\triangleq\frac{\log \Delta}{\log\log\Delta}\;.$
\begin{theorem}\label{thm:iterations}
Fix some $\alpha>1$, the number of iterations of Algorithm~$\alg$
is
\begin{align*}
  O\parentheses{\log_\alpha \Delta + f\cdot \frac{\alpha}{\beta}}
\end{align*}
%$O\parentheses{\frac {f^2} {\epsilon} \cdot \frac{\log \Delta}{\log \log \Delta}}$.
%$$
%O\parentheses{\frac{f\cdot\log\Delta}{\log\log\Delta} + \frac{f^2\cdot\log\log\Delta\cdot(\log\Delta)^{0.001}}{\epsilon}}
%$$
\end{theorem}
\begin{proof}
Fix an edge $e$. We bound the number of iterations until $e$ is covered as follows. Every iteration is either an $e$-raise iteration or a $v$-stuck iteration for some $v\in e$. 
Since $e$ contains at most $f$ vertices, we conclude that the number of iterations is bounded by the number of $e$-stuck iterations plus the sum over $v\in e$ of the number of $v$-stuck iterations.  The theorem follows from Lemmas~\ref{lemma:regular passes} and~\ref{lemma:good}.
\end{proof}
\medskip\noindent
Finally, by setting $\alpha$ appropriately, we bound the running time as follows.
\begin{corollary}
If $\alpha=\log \Delta / \log \log \Delta$, then the round complexity of Algorithm \WHVC is $O\parentheses{\frac{f^2}{\eps}\cdot\frac{\log \Delta}{\log \log \Delta}}$.
\end{corollary}
A refined assignment of $\alpha$ that leads to a reduced dependency of the running time on $\eps$ and $f$ is presented in Section~\ref{sec:refine} in the Appendix.

\bibliographystyle{alpha}
\bibliography{paper}

\begin{thebibliography}{BCGS17}

\bibitem[{\AA}S10]{AstrandS10}
Matti {\AA}strand and Jukka Suomela.
\newblock Fast distributed approximation algorithms for vertex cover and set
  cover in anonymous networks.
\newblock In {\em {SPAA} 2010: Proceedings of the 22nd Annual {ACM} Symposium
  on Parallelism in Algorithms and Architectures, Thira, Santorini, Greece,
  June 13-15, 2010}, pages 294--302, 2010.

\bibitem[BCGS17]{Bar-YehudaCGS17}
Reuven Bar{-}Yehuda, Keren Censor{-}Hillel, Mohsen Ghaffari, and Gregory
  Schwartzman.
\newblock Distributed approximation of maximum independent set and maximum
  matching.
\newblock In {\em {PODC}}, pages 165--174. {ACM}, 2017.

\bibitem[BCS17]{Bar-YehudaCS17}
Reuven Bar{-}Yehuda, Keren Censor{-}Hillel, and Gregory Schwartzman.
\newblock A distributed {(2} + {\(\epsilon\)})-approximation for vertex cover
  in o(log {\(\Delta\)} / {\(\epsilon\)} log log {\(\Delta\)}) rounds.
\newblock {\em J. {ACM}}, 64(3):23:1--23:11, 2017.

\bibitem[BEKS18]{BEKS18}
R.~{Ben-Basat}, G.~{Even}, K.-i. {Kawarabayashi}, and G.~{Schwartzman}.
\newblock {A Deterministic Distributed $2$-Approximation for Weighted Vertex
  Cover in $O(\log n\log\Delta / \log^2\log\Delta)$ Rounds}.
\newblock In {\em {SIROCCO}}, 2018.

\bibitem[CKP16]{ChangKP16}
Yi{-}Jun Chang, Tsvi Kopelowitz, and Seth Pettie.
\newblock An exponential separation between randomized and deterministic
  complexity in the {LOCAL} model.
\newblock In {\em {FOCS}}, pages 615--624. {IEEE} Computer Society, 2016.

\bibitem[CV86]{ColeV86}
Richard Cole and Uzi Vishkin.
\newblock Deterministic coin tossing with applications to optimal parallel list
  ranking.
\newblock {\em Information and Control}, 70(1):32--53, 1986.

\bibitem[EGM18]{unweightedHVC}
Guy Even, Mohsen Ghaffari, and Moti Medina.
\newblock {Distributed Set Cover Approximation: Primal-Dual with Optimal
  Locality}.
\newblock In {\em {DISC}}, 2018.

\bibitem[Gha16]{ghaffari2016improved}
Mohsen Ghaffari.
\newblock An improved distributed algorithm for maximal independent set.
\newblock In {\em Proceedings of the twenty-seventh annual ACM-SIAM symposium
  on Discrete algorithms}, pages 270--277. Society for Industrial and Applied
  Mathematics, 2016.

\bibitem[GS17]{GhaffariS17}
Mohsen Ghaffari and Hsin{-}Hao Su.
\newblock Distributed degree splitting, edge coloring, and orientations.
\newblock In {\em {SODA}}, pages 2505--2523. {SIAM}, 2017.

\bibitem[Kar72]{Karp72}
Richard~M. Karp.
\newblock Reducibility among combinatorial problems.
\newblock In {\em Proceedings of a symposium on the Complexity of Computer
  Computations, held March 20-22, 1972, at the {IBM} Thomas J. Watson Research
  Center, Yorktown Heights, New York.}, pages 85--103, 1972.

\bibitem[KMW06]{KuhnMW06}
Fabian Kuhn, Thomas Moscibroda, and Roger Wattenhofer.
\newblock The price of being near-sighted.
\newblock In {\em Proceedings of the Seventeenth Annual {ACM-SIAM} Symposium on
  Discrete Algorithms, {SODA} 2006, Miami, Florida, USA, January 22-26, 2006},
  pages 980--989, 2006.

\bibitem[KMW16]{KuhnMW16}
Fabian Kuhn, Thomas Moscibroda, and Roger Wattenhofer.
\newblock Local computation: Lower and upper bounds.
\newblock {\em J. {ACM}}, 63(2):17:1--17:44, 2016.

\bibitem[Kuh05]{Kuhn2005}
Fabian Kuhn.
\newblock {\em The price of locality: exploring the complexity of distributed
  coordination primitives}.
\newblock PhD thesis, {ETH} Zurich, 2005.

\end{thebibliography}
\appendix
\section{Primal-Dual Approach}
 \label{sec:primal dual}

\medskip
\noindent
The fractional LP relaxation of \WHVC\ is defined as follows.
\begin{alignat*}{2}\label{eq:P}
  & \text{minimize: }  \sum_{v\in V}w(v) \cdot x(v) &&\\
  & \text{subject to: }& \quad & \\
&\tag{$\mathcal{P}$}
  \begin{aligned}
                   \sum_{v\in e} x(v)  &\geq 1,& \forall e\in E\\%[3ex]
                   x(v) &\geq 0,&\forall v\in V
                 \end{aligned}
\end{alignat*}
The dual LP is an \emph{Edge Packing} problem defined as follows:
\begin{alignat*}{2}\label{eq:D}
  & \text{maximize: }  \sum_{e\in E} \delta(e)&&\\
  & \text{subject to: }& \quad &\\
&\tag{$\mathcal{D}$}
  \begin{aligned}
                   \sum_{e\ni v} \delta(e) &\leq w(v),& \forall v\in V\\
                   \delta(e) &\geq 0,&\forall e\in E
                 \end{aligned}
\end{alignat*}
%%%%%%%%%%%%%%%%%%%%%%%%%%
The following claim is used for proving the approximation ratio of
the \WHVC algorithm.
\begin{claim}\label{claim:ratio PCWHVC}
  Let $\opt$ denote the value of an optimal fractional solution of the
  primal LP~\eqref{eq:P}. Let $\{\delta(e)\}_{e\in E}$ denote a
  feasible solution of the dual LP~\eqref{eq:D}.  Let $\eps\in(0,1)$
  and $\beta\triangleq \eps/(f+\eps)$. Define the $\beta$-tight 
 vertices by:
  \begin{align*}
    T_{\eps} &\triangleq \{ v\in V \mid \sum_{e\ni v} \delta(e)\geq (1-\beta)\cdot w(v)\}.
  \end{align*}
Then $w(T_{\eps})\leq (f+\eps) \cdot \opt$.
\end{claim}
\begin{proof}
  \begin{align*}
    w(T_{\eps}) &=\sum_{v\in T_{\eps}} w(v)\\
                              &\leq \frac{1}{1-\beta} \cdot \parentheses{\sum_{v\in T_{\eps}} \sum_{e\ni v} \delta(e)}\\
                              &\leq \frac{f}{1-\beta} \sum_{e\in E} \delta(e) \leq (f+\eps)\cdot \opt.
  \end{align*}
  The last transition follows from $f/(1-\beta)=f+\eps$ and by  weak 
  duality. The claim follows.
\end{proof}
\section{Adaptation to the CONGEST model}\label{sec:congest}
To complete the discussion, we need to show that the message lengths
in Algorithm \WHVC are $O(\log n)$.
\begin{enumerate}
\item In round $0$, every vertex $v$ sends its weight $w(v)$ and
  degree $|E(v)|$ to every hyperedge in $e\in E(v)$. We assume that
  the weights and degrees are polynomial in $n$, hence the length of the binary
  representations of $w(v)$ and $|E(v)|$ is $O(\log n)$.
  
  Every hyperedge $e$ sends back to every $v\in e$ the pair $(w(v_e),|E(v_e)|)$,
  where $v_e$ has the smallest normalized weight, i.e.,
  $v_e=\argmin_{v\in e} \{w(v)/|E(v)|\}$.

  Every vertex $v\in e$ locally computes $\deal_0(e)=\beta \cdot w(v_e)/|E(v_e)|$
  and $\delta_0(e)=\deal_0(e)$.

\item In round $i\geq 1$, the following types of messages are sent:
  ``$e$ is covered'', ``raise'', or ``stuck''. These messages require
  only a constant number of bits. The decision whether
  $\deal_i(e)=\deal_{i-1}(e)$ or
  $\deal_i(e)=\alpha\cdot \deal_{i-1}(e)$ requires a single bit.
\item Finally, if $\alpha=\alpha(e)$ is set locally based on the local
  maximum degree $\max_{v\in e} |E(v)|$, then every vertex $v$ sends
  its degree to all the edges $e\in E(v)$. The local maximum degree
  for $e$ is sent to every vertex $v\in V$, and this parameter is used
  to compute $\alpha(e)$ locally.
\end{enumerate}
\section{Improved Running Time}\label{sec:refine}
In this section, we present a modified definition of the multiplier $\alpha$ that leads to an improved dependence of the running time on $f$ and $\epsilon$.

\medskip\noindent
Let $\gamma\in (0,1)$ denote a constant.
Set the multiplier $\alpha$ as follows:
\begin{align}\label{eq:alpha}
\alpha \triangleq 
\begin{cases}
\parentheses{
\frac{\log \Delta}{\log\log\Delta}}^{(1-\gamma)}
& \text{if } 
    \frac{f}{\beta}< \parentheses{\frac{\log \Delta}{\log\log \Delta}}^{\gamma}\\
    2 & \text{otherwise.}
\end{cases}
\end{align}

Note that in the following, the round complexity is monotonically nonincreasing in $\gamma$, so it may be chosen arbitrarily close to $1$.
\begin{theorem}\label{thm:refined}
For every constant $\gamma\in (1,0)$, by setting $\alpha$ according to Eq.~\ref{eq:alpha}, 
the round complexity of Algorithm \WHVC is bounded by
\begin{align}\label{eq:improved bound}
    O\parentheses{
    \frac{\log \Delta}{\log\log\Delta}
%    +
%    \frac{f^2}{\eps}
    +
    \parentheses{\frac{f^2}{\eps}}^{1/\gamma} \cdot \log\log \Delta
    }
\end{align}
\end{theorem}
\begin{proof}
By Theorem~\ref{thm:iterations}, the number of iterations is bounded by $\log_\alpha \Delta + f\cdot \frac{\alpha}{\beta}$.
We \mbox{consider two cases.}
\begin{enumerate}
    \item Suppose that 
$\frac{f}{\beta}< \parentheses{\frac{\log \Delta}{\log\log \Delta}}^{\gamma}$.
In this case, $\alpha=\parentheses{
\frac{\log \Delta}{\log\log\Delta}}^{(1-\gamma)}$.
The terms in the bound on the number of iterations satisfy:
\begin{align*}
    \log_\alpha \Delta &= O \parentheses{\frac{\log \Delta}{\log \log \Delta}}\\
    f\cdot\frac{\alpha}{\beta}&\leq \frac{\log \Delta}{\log \log \Delta}\;.
\end{align*}

\item Suppose that
$\frac{f}{\beta}\geq  \parentheses{\frac{\log \Delta}{\log\log \Delta}}^{\gamma}$.
In this case $\alpha=2$, and hence
\begin{align*}
    \log_\alpha \Delta &\leq \parentheses{\frac{f}{\beta}}^{1/\gamma} \cdot \log \log \Delta\\
    f\cdot\frac{\alpha}{\beta}&\leq O
    \parentheses{\parentheses{\frac{f}{\beta}}^{1/\gamma}}\;.
\end{align*}
\end{enumerate}
In both cases, the bound on the number of iterations is bounded by the expression in Eq.~\ref{eq:improved bound}, and the theorem follows.
\end{proof}
\end{document}